\documentclass[11pt]{article}
\usepackage{amscd}
\usepackage{amsthm}
\usepackage{longtable}
\usepackage{mathrsfs}
\usepackage{amsmath}
\usepackage{amssymb}
\usepackage{enumitem}
\newtheorem{thm}{Theorem}
\newtheorem{lem}{Lemma}

\newtheorem{rmk}{Remark}
\newtheorem{cor}{Corollary}
\newtheorem{empl}{Example}
\newtheorem{cnj}{Conjecture}

\begin{document}
\title{Recursions for quadratic rotation symmetric functions weights}
\author{Thomas W. Cusick\\
University at Buffalo, 244 Math. Bldg., Buffalo, NY 14260\\
email cusick@buffalo.edu}
\vspace{.5cm}
\date{}
\maketitle
\begin{abstract}
A Boolean function in $n$ variables is rotation symmetric (RS) if it is invariant under  powers of 
$\rho(x_1, \ldots, x_n) = (x_2, \ldots, x_n, x_1)$. An RS function is called monomial rotation symmetric (MRS) if it is generated by applying powers of 
$\rho$ to a single monomial.  The author showed in $2017$ that for any RS function $f_n$ in $n$ variables, the sequence of Hamming weights $wt(f_n)$ for all values of $n$ satisfies a linear recurrence with associated recursion polynomial given by the minimal polynomial of a {\em rules matrix}. Examples showed that the usual formula for the weights $wt(f_n)$ in terms of powers of the roots of the minimal polynomial always has simple coefficients. The conjecture that this is always true is  the Easy Coefficients Conjecture (ECC). The present paper proves the ECC if the rules matrix satisfies a certain condition. Major applications include an enormous decrease in the amount of computation that is needed to determine the values of $wt(f_n)$ for a quadratic RS function $f_n$ if either $n$ or the order of the recursion for the weights is large, and a simpler way to determine the Dickson form of $f_n.$  The ECC also enables rapid computation of generating functions which give the values of $wt(f_n)$ as coefficients in a power series.
\end{abstract}

recursions, 
generating functions


\section{Introduction} \label{intro}
For an account of basic definitions in the theory of Boolean functions, see \cite[Section 2.1]{CS17}. A more comprehensive reference which also includes coding theory is \cite{Carl20}. A Boolean function $f=f_n$ in $n$ variables is said to be {\it rotation symmetric} (RS) if the {\it algebraic normal form} \cite[p. 8]{CS17} (polynomial form) of the function is unchanged by any cyclic permutation of the variables
$\rho(x_1,\ldots ,x_n)=(x_2,\ldots,x_n,x_1)$. The {\it degree} of $f$ (notation $\deg(f)$) is the degree of this polynomial (that is, the number of variables $x_i$ which appear in the longest monomial in the algebraic normal form of $f$). If an RS function is generated by a single monomial, we say that the function is {\it monomial rotation symmetric} (MRS). Let $V_n$ denote the vector space of dimension $n$ over the finite field $GF(2).$ The (Hamming) {\it weight} of $f,$ denoted by $wt(f_n)$ or $wt(f),$ is number of digits $1$ in  the set  $\{f(x): x \in V_n\};$  this set is called the {\it truth table} of $f.$ If a Boolean function $f(x)$ in $n$ variables has weight $2^{n-1},$ we say the function is \emph{balanced}.
We say that two Boolean functions $f({x})$ and $g({x})$ in $n$ variables are {\it affine equivalent}  if $g({x}) = f(A{x} +  {b})$, where $A$ is an $n$ by $n$ nonsingular matrix over the finite field $GF(2)$ and $b \in V_n.$  

The rotation symmetric name was introduced in \cite{PQ}, which showed that RS functions could be used in the design of cryptographic hash functions.
There are some earlier papers such as \cite{FF} which studied cryptographic applications of  RS functions but under the label \emph{idempotents}.  It turns
out that the set of RS functions of various degrees is extremely rich in functions with useful cryptographic properties.  Some papers illustrating this include
\cite{CG, InfSci11, NTT, IET, kavut3, kavut4, kavut5, kmy, kavut, kavut2}. A very recent paper on using RS functions to design cryptosystems is \cite{SS}. More references and details of the uses of RS functions are in \cite[Section 6.2]{CS17} and \cite{Zv}.  In particular, the article \cite{Zv} points out that the functions used to prove the famous lower bound of $16,276$ for the covering radius of the $(2^{15},16)$ Reed-Muller code (see \cite{PW}) are rotation symmetric, and discusses further work arising from \cite{PW}.

We write $(1,t)_n$ for the quadratic MRS function in $n$ variables generated by the monomial $x_1x_t$, so the algebraic normal form is
\begin{equation}\label{eq1}
(1,t)_n=x_1 x_t + x_2 x_{t+1} + \ldots + x_n x_{t-1}
\end{equation}

In \cite{CC20, CC22}, it was convenient to label the $n$ variables $x_0, x_1, \ldots, x_{n-1},$ so the function in \eqref{eq1}  would be represented as $(0,t-1)_n.$ We will not use that notation in this paper. The reader should remember this when we refer to results in those papers (for instance, in Lemma \ref{1t} and Theorems \ref{thmold} and \ref{vn} below).

An important tool in Boolean function theory (see \cite[Section 2.2]{CS17}; notice our notation is a little different) is the {\it Walsh transform}  $W_f(u)$ for $f=f_n$ in $n$ variables and $u$ in $V_n.$ It is defined by 
\begin{equation} \label{Walsh}
W_f(u)= \sum_{x \in V_n} (-1)^{f(x) + u \cdot x}.
\end{equation}  In this paper we shall mainly need the case $u = \bf{0}$ which gives the simple formula
\begin{equation*}
\sum_{x \in V_n} (-1)^{f(x)} = 2^n - 2wt(f).
\end{equation*}

A Boolean function of degree $\leq 1$ is called {\it affine}, and if the constant term is $0$ the function is {\it linear}.  The {\it nonlinearity} $N(f)$  of any Boolean function is defined by 
$N(f) = \min_{\text{a affine}} ~wt(f + a).$   The following lemma (well-known to experts but apparently there was no published proof until \cite[Lemma 2.3, p. 5068]{InfSci11}) shows that affine equivalence for any quadratic Boolean function $f$ is completely determined by $wt(f)$ and $N(f).$ This is not true for any degree $>2.$

\begin{lem} \label{lem1}
Two quadratic Boolean functions $f$ and $g$ in $n$ variables are affine equivalent if and only if
$wt(f) = wt(g)$ and $N(f) = N(g).$
\end{lem}

There is also a well-known  \cite[Theorem 2.21]{CS17} result giving the nonlinearity in terms of values of the Walsh transform:
\begin{lem} \label{lem2}
For any Boolean function $f$ in $n$ variables we have
\begin{equation}
N(f)= 2^{n-1} -\frac{1}{2} \max_{x \in V_n} |W_f(x)|.
\end{equation}
\end{lem}

After earlier work on special cases, the following general theorem was proved in \cite{C18}. 
\begin{thm} \label{rec18}
Let $f_k(x_1, x_2, \ldots, x_k)$ denote any rotation symmetric function in $k$ variables equal to the sum of $m$ monomials $x_1x_{a_2}(i) \cdots x_{a_{d(i)}}(i),~1 \leq i \leq m$, where $d(i)$ is the degree of the $i$-th monomial. Then the sequence of weights $\{w_k=wt(f_k): k \geq max~d(i)\}$ satisfies a homogeneous linear recursion with integer coefficients.  The recursion polynomial is the minimal polynomial for a square  {\emph rules matrix} (which depends on the choice of the function $f_k$), with any powers of $x$ in the minimal polynomial removed. 
\end{thm}
An algorithm for computing the rules matrix was explained in \cite{C18}. A Mathematica program for computing the rules matrix  was given in \cite{Carx}. This program clearly shows why the minimal polynomial of the rules matrix gives the recursion polynomial for the weights.  The basic idea goes back to \cite[Theorem 1, p. 116]{bc12}, where the very special case of cubic MRS functions was worked out in detail.  Reading \cite{C18} requires familiarity with the work in \cite{bc12}. Both the computation of the rules matrix and the proof that the weight recursions exist are complicated, but applying those results to any given Boolean function by using the program in \cite{Carx} is straightforward.  

The quadratic case was not explicitly considered in \cite{C18}, since the $2009$ paper \cite{Kim} had already shown that, in principle, the weights for quadratic $f_n$ could be given by an explicit formula, without using weight recursions. However, this formula is not easy to compute for a typical quadratic function because the proof depends on the well-known old theorem of Dickson \cite[Theorem 4]{Kim}. Using the theorem of Dickson requires computing a particular special quadratic function equivalent to the one being considered (this is called the \emph{Dickson form} of the quadratic; see
 \cite[pp. 172-173]{Carl20}) and the computation is very difficult as the number of variables increases. A full proof of Dickson's theorem and the needed computation is in \cite[pp. 438-442]{Sloane}.  Also the paper \cite{CC20} shows that the weight recursions for the quadratic functions can be used to obtain much deeper results, for example deciding exactly which quadratic RS functions are balanced \cite[Theorem 5.3]{CC20}. Using the method of \cite{C18} for the quadratic case is significantly easier than the higher degree cases described in detail in \cite{C18}. A detailed description of how the  rules matrix can be computed (it could be regarded as pseudocode) in the quadratic case is given in the Appendix. Since the weight of the function determines the Dickson form, computing the rules matrix provides a simpler method for determining the Dickson form.

Next is an
example of the rules matrix for the function $f_n = (1, 2)_n$ ($m=1$ in Theorem \ref{rec18}), plus the resulting weight recursion and the weights $wt(f_n).$
\begin{empl} \label{ex1}
The rules matrix for $f_n = (1,2)_n$ has size $(3,3)$ and its rows are $$\{1,1,0\}, \{1, -1, 0\}, \{0,1,2\}.$$  The minimal and characteristic polynomials for the rules matrix are the same, namely $x^3-2x^2-2x+4 = (x-2)(x^2-2).$ The weights $w_n = w_n((1,2)_n)$ satisfy the recursion  $w_{n+3}= 2w_{n+2} + 2w_{n+1} -4w_n,$ whose coefficients are given by the minimal polynomial, for $n \geq 1.$  Of course $(1,2)_n$ only is defined if $n \geq 2,$ but we can extend the recursion backwards to get $w_n$ for
$1 \leq n < deg ~f_n.$   Then the first $10$ values $wt( (1,2)_n), n \geq 1$ are $1, 0, 4, 4, 16, 24, 64, 112, 256, 480.$  The value $wt((1,2)_2) = 0$ occurs because the
MRS functions $(1,t)_{2t-2}$ are \emph{short}, that is \eqref{eq1} contains two copies of each monomial for these functions and so they are the $0$ function if all
$2t-2$ monomials are included.  In fact, it is usual to define $(1,t)_{2t-2} = x_1x_t + \ldots + x_{t-1}x_{2t-2}$ (considering only the first $t-1$ monomials in \eqref{eq1}),
so we would obtain the weight $2^{2t-3} - 2^{t-2}$ and would say $(1,t)_{2t-2}$ is a \emph{bent} function (see \cite[Table 1, p. 430]{Kim}).  The recursions in this paper require us to regard these bent functions as identically $0.$   
\end{empl}
As we know from Theorem \ref{rec18}, the integer coefficients for the weight recursion for  $f_n$ in Example \ref{ex1} are the coefficients of the minimal polynomial for the rules matrix for $f_n$.  It turns out that for quadratic functions, this minimal polynomial is never divisible by $x.$
\section{The rules matrix for quadratic RS functions} 

Below we give a brief description of the quadratic case of the algorithm in \cite{C18}. The corresponding Mathematica program is much shorter and simpler than the general one in \cite{Carx}.   Many commands in \cite{Carx} are not needed in the quadratic case, and can simply be deleted. For the convenience of readers, especially those who do not have access to the Mathematica software, the Appendix contains a full description of how the algorithm for computing the rules matrix is carried out.

We begin with a set 
\begin{equation} \label{Qm}
 Q_m = \{(1,t_1), (1,t_2), \ldots, (1,t_m)\}
 \end{equation}
of quadratic MRS functions and we shall describe the program for computing the rules matrix for the sum of those $m$ functions. We let $f(Q_m)$ denote the sum of the functions in \eqref{Qm}, or $f^{(m)}$ for short if \eqref{Qm} is is assumed to be fixed. We let $r(f^{(m)}) = r(f(Q_m))$ denote the rules matrix.  The following remark summarizes some features of the rules matrix which the program produces. Since the last row and column of the rules matrix have an especially simple form different from all other rows and columns, it is convenient to define the \emph{Rules matrix} to be the rules matrix with the last row and the last column omitted, and to let $R(f^{(m)}) = R(f(Q_m))$ denote the Rules matrix.
\begin{rmk} \label{rk1}
The rules matrix is square with $2^k+1$ rows (we say it has size $2^k+1$), where 
$k= max~(t_i-1).$
Every row and column in the Rules matrix has exactly two nonzero elements, and these nonzero elements are always either $1$ or $-1.$  No row or column contains two elements $-1.$ The last column of the rules matrix is always $[0,0,\ldots,0,2]$
and the last row has  $1$ in every column which contains $-1$, with all other entries $0$ in the last row.
\end{rmk}

Theorem \ref{thm2} below gives every statement in Remark \ref{rk1}, and more.
Here we briefly describe the program which computes the rules matrix (for details see the Appendix).  The program first computes a much larger (size $2^a+1$ with $a = \sum_{i=1}^n (t_i-1)$) matrix with many zero rows. We call this the 
${\emph expanded~rules~matrix}.$ Once the expanded matrix has been computed, the next step in the program is to delete the $i$-th row if it is all zeros and to also delete the $i$-th column. If the resulting matrix still has any zero rows, the process of deleting zero rows and the corresponding columns is repeated  until a matrix with no zero rows is reached.

The next lemma gives some important facts about the expanded rules matrix.
\begin{lem} \label{expnd}
The expanded rules matrix $E(f(Q_m))=E$ for a set $Q_m$ given by \eqref{Qm} has size $2^u+1$ where $u= \sum_{i=1}^m (t_i-1).$ Every row in $E$ except the last one has either $0$ or $2^m$ nonzero elements, and the nonzero elements are always $\pm 1.$  Every column in $E,$ ignoring the last entry,  has exactly two nonzero entries, which are either $1,1$ or $1,-1.$ The last column of $E$ is always $[0,0,\ldots,0,2]$
and the last row has  $1$ in every column which contains $-1$, with all other last row entries $0.$ After each pass through $E$ in which all zero rows and the corresponding columns are deleted, the smaller matrix has the same characteristic and minimal polynomials as $E.$  The final rules matrix with no zero rows has nonzero determinant and size $2^k+1$ where  $k = \text{max} ~(t_i-1).$ Also the last column is always $[0,0,\ldots,0,2],$ the last row has  $1$ in every column which contains $-1$, and all other last row entries are $0.$
\end{lem}
\begin{proof}
All of the assertions follow from the algorithm for computing the rules matrix--see the Appendix.
\end{proof}

The next theorem gives some important properties of the Rules matrices. Note that the Rules matrix has an especially simple form if 
\begin{equation} \label{ti}
t_1>t_2> \ldots >t_m
\end{equation}
is true (as in Example \ref{ex3} below).
\begin{thm} \label{thm2}
Every row and every column in the Rules matrix for any sum of quadratic RS functions 
(say the functions are given by $Q_m$ in \eqref{Qm}) has exactly two nonzero entries, 
which are either $1, 1$ or $1, -1$.  There is an equal number of rows with two $1$'s and with a $-1.$  If a row with $1,1$ is given, there is another row 
with $1,-1$ or $-1, 1$ in the same columns in which the two $1$’s appear. Also if
a row with entries $1,-1$ is given, there is another row with $1,1$ in the same columns 
in which the $1$ and $-1$ appear.  We call the two rows with nonzero entries in the
same columns {\em paired~rows}. If the functions in $Q_m$ satisfy \eqref{ti}, then any two paired rows
in the Rules matrix are consecutive, with the first nonzero element in row $2i-1, ~1 \leq i \leq 2^{k-1},$
where $k = \text{max} ~(t_i-1),$ occurring in position $i.$  Every row has one of its two nonzero entries in its first half and the other in its
second half, with the entries in each half the same distance away from the first element in the respective
halves
\end{thm}
\begin{proof}
All the assertions follow from Lemma \ref{expnd} and closely examining the detailed description of the computation of the rules matrix in the Appendix. Equations \eqref{E1},  \eqref{E2} and  \eqref{E3} are essential for this. Notice that the  left and right halves
of the Rules matrix are computed independently.
\end{proof}
\begin{cor} \label{cor1}
If $A$ denotes the Rules matrix, then $A \cdot A^T = 2I,$ where $T$ denotes transpose and $I$ is the identity matrix of the appropriate size.
Therefore $(1/\sqrt{2})A$ is an orthogonal matrix and is diagonalizable over $\mathbf{C}.$ Hence the minimal polynomial of $A$ has distinct nonzero roots, which are the eigenvalues of $A.$  Any matrix of the same size as $A$ with the same set of eigenvalues is permutation equivalent to $A.$
\end{cor}
\begin{proof}
The first sentence follows from the detailed description in Theorem \ref{thm2} of the placement of all the nonzero entries in the Rules matrix. Then the 
fact that $(1/\sqrt{2})A$ is an orthogonal matrix is immediate. It is a standard linear algebra result that a real orthogonal matrix is diagonalizable
over $\mathbf{C}.$ The final assertions in the lemma follow from the fact that a matrix $M$ is diagonalizable over $\mathbf{C}$ if and only if the minimal polynomial of $M$ splits over $\mathbf{C}$ and is squarefree. 
\end{proof}

If \eqref{ti} holds, Theorem \ref{thm2} implies that the Rules matrix always has two diagonal stripes made up of paired nonzero entries,
one stripe in each of the left and right halves of the matrix. Looking at the following examples will be helpful in
understanding the content of Theorem \ref{thm2}, and also some other comments later in the paper. Note that in the examples all paired rows are consecutive, as guaranteed by Theorem \ref{thm2}.
\begin{empl} \label{ex3}
Let $Q_3=\{(1,4), (1,3), (1,2)\}.$  Then  \eqref{ti} is true and the rules matrix (size $9$) is
\begin{equation*}
   \left(
   \begin{matrix}
    1& 0  & 0  & 0  &  1 & 0  & 0  & 0  & 0 \\
    1& 0  & 0  & 0  & -1 & 0  & 0  & 0  & 0 \\
    0& 1  & 0  & 0  &  0 & 1  & 0  & 0  & 0 \\
    0&-1  & 0  & 0  &  0 & 1  & 0  & 0  & 0 \\
    0& 0  & 1  & 0  &  0 & 0  & 1  & 0  & 0 \\
    0& 0  &-1  & 0  &  0 & 0  & 1  & 0  & 0 \\
    0& 0  & 0  & 1  &  0 & 0  & 0  & 1  & 0 \\
    0& 0  & 0  & 1  &  0 & 0  & 0  &-1  & 0 \\
    0& 1  & 1  & 0  &  1 & 0  & 0  & 1  & 2  
 \end{matrix}
 \right)
\end{equation*}
The minimal and characteristic polynomials for this matrix are both $$x^9-2x^8-4x^5+8x^4+16x-32.$$
\end{empl}

The next example shows that the minimal and characteristic polynomials for the rules matrix can be different, even for $m=1$ in \eqref{Qm}.
\begin{empl} \label{ex4}
Let $Q_1=\{(1,4)\}.$  Then the rules matrix (size $9$) is
\begin{equation*}
   \left(
   \begin{matrix}
    1& 0  & 0  & 0  &  1 & 0  & 0  & 0  & 0 \\
    1& 0  & 0  & 0  & -1 & 0  & 0  & 0  & 0 \\
    0& 1  & 0  & 0  &  0 & 1  & 0  & 0  & 0 \\
    0& 1  & 0  & 0  &  0 &-1  & 0  & 0  & 0 \\
    0& 0  & 1  & 0  &  0 & 0  & 1  & 0  & 0 \\
    0& 0  & 1  & 0  &  0 & 0  &-1  & 0  & 0 \\
    0& 0  & 0  & 1  &  0 & 0  & 0  & 1  & 0 \\
    0& 0  & 0  & 1  &  0 & 0  & 0  &-1  & 0 \\
    0& 0  & 0  & 0  &  1 & 1  & 1  & 1  & 2  
 \end{matrix}
 \right)
\end{equation*}
The minimal polynomial for this matrix is $$x^7-2x^6-8x+16 = (x-2)(x^2-2)(x^4+2x^2+4)$$
and the characteristic polynomial is $$x^9-2x^8-2x^7+4x^6-8x^3+16x^2+16x-32 = (x-2)(x^2-2)^2(x^4+2x^2+4).$$
\end{empl}

\section{The quadratic Easy Coefficients Conjecture}
We begin with some results that are easily derived from Theorem \ref{thm2} and the details in the Appendix.
First, we can extend Example \ref{ex1} to any MRS function $(1,t)_n.$

\begin{lem} \label{1t}
The minimal polynomial $m_t(x)$ for the rules matrix (which is square of size $2^{t-1}+1$) of $(1,t)_n$ is
$$m_t(x) = x^{2t-1}-2x^{2t-2}-2^{t-1} x+2^{t}= (x-2)(x^{2t-2} - 2^{t-1}).$$
The characteristic polynomial $c_t(x)$ for the rules matrix has degree $2^{t-1} +1,$ has no $0$ roots and the root $2$ is never a multiple root.
\end{lem}
\begin{proof}
A direct proof of this is given in \cite[Theorems 3.1 and 3.4, pp. 1306]{CC20}.  Note that the rules matrix defined in this paper has an extra final row and column (corresponding to the root $2$ which the minimal polynomial of the rules matrix always has) compared to the "rules matrix" (called the Rules matrix in this paper) defined in \cite{CC20}. 
This lemma can also be deduced from the work in the Appendix.
\end{proof}

We see from Lemma \ref{1t} that the minimal polynomial for the rules matrix of $(1,t)_n$ in general has much smaller degree ($2t-1$)  than the characteristic polynomial (degree $2^{t-1}+1$).  A similar result is true for the general quadratic RS function $f(Q_m)=f^{(m)}$ defined by \eqref{Qm} when \eqref{ti} holds. The next lemma shows that changing the order of the functions in $Q_m$ simply replaces the rules matrix by a similar matrix, and in fact the two rules matrices are permutation equivalent.
\begin{lem} \label{perm}
Let $R(f^{(m)})$ denote the Rules matrix for  $f^{(m)}=f(Q_m)$ when $n>1$ and \eqref{ti} holds. If $R'(f^{(m)})$ denotes the Rules matrix for $f^{(m)}$ with some different ordering of the integers $t_i,$ then $R'(f^{(m)})= PR(f^{(m)})P^{-1}$ for some permutation matrix $P$ (which is the identity matrix $I$ with permuted rows). Obviously this also holds for the rules matrix.
\end{lem}
\begin{proof}
It follows from Theorem \ref{thm2} and \eqref{E1},  \eqref{E2} and  \eqref{E3} that the number of zero rows in the expanded rules matrix for $f^{(m)}$
does not depend on the ordering of the functions in \eqref{Qm}.  It also follows that there is a one-to-one correspondence between the columns in $R(f^{(m)})$ with two entries $1$ and the columns of $R'(f^{(m)})$ with two entries $1,$  and a similar correspondence between the columns with entries $1$ and $-1.$  Therefore, whatever the ordering of the functions in $Q_m,$ the Rules matrix which is produced by the process of deleting all the zero rows from 
$R'(f^{(m)})$ is permutation equivalent to $Rf^{(m)}).$ The final sentence in the lemma is trivially true because of the special form (see Lemma \ref{expnd}) of the last row and column of the rules matrix.
\end{proof}

The method of \cite{C18, Carx}, as greatly simplified in the Appendix, shows how to compute the rules matrix.  This leads to the following theorem.
\begin{thm} \label{char}
For any quadratic RS Boolean function $f^{(m)} = f(Q_m)$,  where $Q_m$ is given by \eqref{Qm}, let $T = \max~ (t_i-1)$ and let $r(f^{(m)})$ denote the rules matrix for $f^{(m)}.$  The method in the Appendix gives the nonsingular square rules matrix of size $2^{T}+1.$ Let $m_f(x)$ denote the minimal polynomial of $r(f^{(m)})$ and let $c_f(x)$ denote the characteristic polynomial of $r(f^{(m)}).$ Then the weights of $f_n$ satisfy a linear recursion with integer coefficients of order $\deg m_f(x).$ The weights also satisfy a linear recursion with integer coefficients of order $\deg c_f(x)=2^T+1.$
\end{thm}
\begin{proof}
Everything but the last sentence is simply the quadratic case of Theorem \ref{rec18}. Note that by Lemma \ref{perm} we can assume without loss of generality that \eqref{ti} holds.  By Corollary \ref{cor1} the minimal polynomial has distinct roots.  It is well-known that the sets of roots of $m_f(x)$ and $c_f(x)$ are the same, but roots (except for $2$) can occur with multiplicity $>1$ in $c_f(x).$  Thus by adjusting the coefficients in the weight recursion of order $\deg m_f(x)$ we can get a recursion of order $\deg c_f(x).$
\end{proof}

Given $f^{(m)} = f(Q_m),$ let the roots of the minimal polynomial $m_f(x)$ of $r(f^{(m)})$ of degree $D=D(f^{(m)})$ be $2$ and  $\mu_i, 1 \leq i \leq D-1.$ We know from Corollary \ref{cor1} that these roots are distinct. Then
from Theorem \ref{char} and the theory of linear recurrences  (see \cite[pp. 1-6]{recur}) it is possible to represent $wt(f_n)$ as a linear combination of the $n$-th powers 
of the roots of $m_f(x)$ with complex number coefficients and it is also possible to represent $wt(f_n)$ as a linear combination of the $n$-th powers of the roots of $c_f(x)$ with some other coefficients.   The Easy Coefficients Conjecture (ECC) stated below says that if we choose the representation in terms of the roots of $c_f(x)$ (this just means repeating each root of $m_f(x)$ as many times as it occurs as a root of $c_f(x)$), then the coefficients are very simple (all are $-\frac{1}{2}$ except for the coefficient of the root $2,$ which is $\frac{1}{2}$ ).
\begin{cnj}  \label{ECC1}
{\bf Quadratic Easy Coefficients Conjecture}\\
For any quadratic RS Boolean function $f^{(m)}=f(Q_m),$ let $T = \max~ (t_i-1).$ Then the method of \cite{C18, Carx} gives a nonsingular square rules matrix $r(f^{(m)})$ of size $2^T+1.$ Let the roots of the minimal polynomial of this matrix be $2$ and  $\mu_i, 1 \leq i \leq D-1.$ Let the multiset  
$\{\eta_i:~1 \leq i \leq 2^T\}$  consist of the roots $\mu_i,$ each repeated as often as its multiplicity in the characteristic polynomial of $r(f^{(m)}).$ Then 
\begin{equation} \label{wtsumA}
      wt(f_n) = 2^{n-1} - \sum_{j=1}^{2^T} \frac{1}{2} \eta_j^n,~n = 1, 2, \ldots
    \end{equation}
\end{cnj}
The ECC implies an enormous decrease in the amount of computation that is need to determine the values of $wt(f_n)$ for a quadratic RS function $f_n$ if either $n$ or the order of the recursion for the weights is large.  Without the ECC, even if we know the coefficients in a recursion polynomial of degree $d$ for the weights $wt(f_n),$  we need to determine the initial $2d+1$ values of $wt(f_n)$ in order to compute any weight values for $n$ large.  This is an exponential computation, bounded in size by $O(n2^n),$ and so not feasible for large $n.$  A detailed discussion of such computations is given in \cite[Section 6]{NTT}  and \cite[Section 3]{ecc24}.  See also \cite[p. 222]{recur}. With the ECC, we need only compute the decimals for the roots $\mu_i$ with very high accuracy,
and this computation is \emph{linear} in $n,$ apart from logarithmic factors (see \cite[Section 6]{NTT} for details). Thus the ECC allows us to reduce an exponential computation  in $n$ to one that is essentially linear in $n.$

Another way to very efficiently determine the values of $wt(f_n)$ is to find a \emph{generating function} $\text{gen}(f_n)(x)$ such that 
\begin{equation} \label{gen}
\text{gen}(f_n)(x) = \sum_{i=1}^{\infty} wt(f_i) x^{i-1}
\end{equation}
Here the values of $wt(f_i)$ for $i < n$ can be obtained by extending the recursion for the weights $wt(f_i)$ backwards from $i=n,$ as in Theorem \ref{thmold} below.

It turns out that the generating functions $\text{gen}(f_n)(x)$ for quadratic RS functions $f_n=f(Q_m)$ are rational functions with integer coefficients, and these rational functions can be easily computed by using the ECC. Full explanations and examples are given in \cite{ecc24}.  We give the generating function for  the Example \ref{ex4} function $f_n = (1,4)_n$ here.
\begin{empl} \label{ex5}
For $f_n = (1,4)_n,$ 
\begin{equation*}
\text{gen}(f_n)(x) = \frac{56x^6-32x^5+8x^4-4x^3+4x^2-2x+1}{16x^7-8x^6-2x+1}
\end{equation*}
and the expansion of this rational function is
$$\text{gen}( (1,4)_n )(x) = 1 +  4x^2 +4x^3 + 16x^4 + 64 x^6 + 112 x^7 + 256 x^8 + 480x^9 \ldots$$
\end{empl}

Computing the weights in Examples \ref{ex3} and \ref{ex4} confirms Conjecture \ref{ECC1} for those cases.
The idea goes back to \cite{NTT} where a weaker version for degree $3$ MRS functions only was proved. The special case of MRS quadratics was done in \cite{CC22}. We state that result as our next theorem; note that in this special case we can give a more detailed result than the  quadratic ECC stated above as Conjecture \ref{ECC1}.
\begin{thm} \label{thmold}
Let $w_n = wt((1,t)_n)$ for $n \geq 2t-1.$ The recursion for the weights is
\begin{equation} \label{wtrec}
w_n = 2w_{n-1} + 2^{t-1}w_{n-2t+2} - 2^{t}w_{n-2t+1},~n \geq 2t-1.
\end{equation}
If we extend the recursion backwards from $n=2t-1$ to $n=1$ to define $w_n$ for $n \leq 2t-1,$ then
\begin{equation} \label{wn}
w_n = 2^{n-1} - \frac{1}{2}(\alpha_1^n + \ldots + \alpha_{2^t}^n),~n = 1, 2, \ldots.
\end{equation}
Here $\alpha_1, \ldots, \alpha_{2^t}$ is the list of the $2^t$ irrational roots, counted with multiplicity, of the characteristic polynomial $c_t(x)$ for the rules matrix $r((1,t)),$ with the distinct roots 
$\alpha_1, \ldots, \alpha_{2t}$ of the polynomial $m_t(x)/(x-2) =x^{2t}-2^t$ for $r((1,t))$ listed first. The remaining roots are various duplicates of the first $2t$ roots.
\end{thm}
\begin{proof}
We already saw the minimal polynomial for this case in Lemma \ref{1t}; this gives \eqref{wtrec}. The full details of the proof of the theorem are in
\cite[Theorems 4.4 and 4.5]{CC22}.
The condition $n \geq 2t-1$ in \eqref{wtrec} is needed because the recursion works only for $n \geq 2t-1.$ This is true because extending the recursion backwards gives $w_{2t-2}=0.$ However, if we adopt the convention that the short functions $(1,t)_{2t-2}$ are identically $0$ (see Example \ref{ex1}),
then \eqref{wtrec} gives $w_n = wt(1,t)_n$ for all $n$ for which $wt((1,t)_n)$ is defined (that is, $n \geq t$).
\end{proof}

Theorem \ref{thmold} proves Conjecture \ref{ECC1} for the case $m=1$ of \eqref{Qm}, that is, the ECC Conjecture for quadratic MRS functions.  The purpose of this paper is to prove an easier version of the ECC Conjecture (a hypothesis is added) for many quadratic RS functions with any value of $m$.  

\section{Special case of  Easy Coefficients Conjecture} \label{EECC}
We shall prove the quadratic ECC in the special case where $m_f(x)=c_f(x)$ in Theorem \ref{char}. This case frequently occurs and is clearly equivalent to saying that $c_f(x)$ has no multiple roots. The proof depends on a deep connection between the weights of $f(Q_m)$ and certain zeta functions,
as explained in \cite{CC22}.  We summarize the results from \cite{CC22} that we need.  These results are in fact true for
rotation symmetric functions of any degree, but in this paper we only need the degree $2$ case. 

Before proceeding with our work, we need some results from \cite{CG, CC20, CC22}. First, it is well known that quadratic Boolean functions are \emph{plateaued} \cite[p. 263]{Carl20}: The weight of a quadratic Boolean function $f_n$ is either $2^{n-1}$ (balanced function) or is of the form $2^{n-1}\pm 2^{\frac{n+v(n)}2-1}$ for some integer $v(n)$ of the same parity as $n$. We use the notation $v_f(n)$ if it is necessary to emphasize the dependence of $v(n)$ on the function $f_n.$ 

Using the results in \cite[Sections 3 and 4]{CC22} and the fact that the values of $v_f(n)$ are periodic (how to compute the period length, $P(v_f)$ say,
is explained in \cite[Theorem 5.2]{CC20} and \cite[Proof, p. 95]{CC24}, but we do not need that information here), we can get even more complete results below.
\begin{thm} \label{countrts}
Given a quadratic function $f_n=f(Q_m)$, there is a  Galois-invariant multiset of algebraic
integers $\alpha_i$ such that
  \begin{equation*}
    |\alpha_i|=\sqrt 2,\ 1\le i,~j\le \max_{n} ~2^{v(n)/2},
  \end{equation*}
 and
  \begin{equation} \label{alpha}
    wt(f_n) = 2^{n-1}-\frac{\sum\alpha_i^n}2.
  \end{equation}
Furthermore, the  order of the group generated by the roots of unity $\alpha_i/\sqrt 2$  is either the period $P(v_f)$ or $2P(v_f)$. 
\end{thm}
\begin{proof}
This is a version of \cite[Theorem 5.2]{CC22}.
\end{proof}

We  find the size of the multiset in Theorem \ref{countrts} in the next theorem, which gives the maximum value for $v_f(n).$
\begin{thm} \label{vn}
Given a quadratic function $f_n=f(Q_m),$ let $T = \max (t_i-1).$  Then $\max_n v_f(n) = 2T,$ so the multiset in Theorem \ref{countrts}  has size $2^T.$
\end{thm}
\begin{proof}
The proof of $\max_n v_f(n) = 2T$ (and more), in a different notation, is given in \cite[Theorem 5.2]{CC20}.
\end{proof}
We can now prove the special case of the quadratic ECC mentioned in the first paragraph of Section \ref{EECC}. We call this
result the Easier Easy Coefficients Theorem (EECT).

\begin{thm}  \label{ECC}
{\bf The Easier Easy Coefficients Theorem}
For any quadratic RS Boolean function $f_n = \sum_{i=1}^m (1, t_i)_n$ let $T = \max~ (t_i-1).$ Then the method of \cite{C18, Carx} gives a rules matrix $r(f_n)$ of size $2^{T}+1$ and a Rules matrix $R(f_n)$ of size $2^T$ such that $(1/\sqrt{2})R(f_n)$ is orthogonal. Let the roots (which are distinct and nonzero) of the minimal polynomial $m_f(x)$ of $r(f_n)$ be $2$ and  
$\{\mu_i: 1  \leq i \leq 2^T\}.$  Then the characteristic polynomial $c_f(x)$ of $r(f_n)$ has the same roots and
\begin{equation} \label{wtsum}
      wt(f_n) = 2^{n-1} - \sum_{i=1}^{2^T} \frac{1}{2} \mu_i^n,~n = 1, 2, \ldots
\end{equation}
The weights $wt(f_n)$ satisfy a linear recursion with integer coefficients of order $\deg m_f(x)=2^T+1.$
\end{thm}
\begin{proof}
Theorem \ref{char} guarantees that the weight recursion in the last sentence of the theorem exists. We know from Corollary \ref{cor1} that
$(1/\sqrt{2})R(f_n)$ is orthogonal, so the roots $\mu_i$ of $m_f(x)/(x-2)$ are distinct and nonzero, and all have absolute value $\sqrt{2}.$ Now
Theorem \ref{countrts}, \eqref{alpha} and Theorem \ref{vn} tell us that there is a Galois invariant multiset of algebraic integers $\alpha_i$ with
$|\alpha_i| = \sqrt{2},~1 \leq i \leq 2^T,$ such that $wt(f_n)$ satisfies \eqref{wtsum} with $\mu_i$ replaced by $\alpha_i.$

The theory of recursions says that the set of distinct algebraic integers $\mu_i,~1 \leq i \leq 2T$ that gives the linear recursion of order $2^T+1$ for
$wt(f_n)$ is uniquely determined by the weights. Since the multiset of the  $\alpha_i$ has the same size $2^T,$ it must in fact be a set ($2^T$ distinct elements).
Therefore the Galois invariant set of algebraic integers $\alpha_i$ must be the same as the set of the $\mu_i$ in \eqref{wtsum}.
This means there is a one-to-one correspondence between the sets made up of the $\alpha_i$ and the $\mu_i,$ so the theorem is proved.   
\end{proof}

The situation where the characteristic polynomial  $c_f(x)$ has multiple roots is more difficult, because then it is not true that the recursion polynomial
for the $wt(f_n)$  recursion of order equal to the degree of the characteristic polynomial is unique.  Here the roots $\mu_i$ of $c_f(x)$ are a multiset
and new ideas are needed to connect that multiset with the $\alpha_i$ in Theorem \ref{countrts}.

\section{Topics for further research}
Many computations confirm that the Easy Coefficients Conjecture (as stated in \cite[Conjecture 1]{ecc24}) is true for RS functions of any degree. The
computation of the weights via generating functions \eqref{gen} can be done just as in the quadratic case, but for degrees $> 2$ there is no nice formula
for the generating functions like the one in \cite[Theorem 2]{ecc24}.

The proof of the Easy Coefficients Conjecture for degrees larger than $2$ will require new ideas, since the higher degree cases are more complicated. For example, we know from Lemma \ref{char} that for a quadratic function $f_n=f(Q_m)$ the rules matrix has size $2^T+1$ and (although the minimal polynomial can have degree $< 2^T+1$) the characteristic polynomial always has degree $2^T+1,$ and has no zero roots. The rules matrix is always nonsingular (that is, has nonzero determinant).

The situation for higher degrees is very different.   Using notation $$(1, c(1), c(2), \ldots, c(d-1))$$ for the MRS function generated by the monomial $x_1 x_{c(1)} \ldots x_{c(d-1)},$ the RS function  $(1, 2, 5) + (1, 5)$ has rules matrix of size $17,$ minimal polynomial of degree $12$ with $8$ nonzero roots and characteristic polynomial of degree $17$ with $8$ nonzero roots. Thus the rules matrix has determinant zero and its characteristic polynomial has  $9$ zero roots. Also, the absolute values of the nonzero roots $\neq 2$ in the minimal polynomial are not the same.  Despite these differences, the ECC is always true for the many examples we have done. It is also the case that rapid computation of generating functions for the weights, as in Example \ref{ex5}, is possible for all of the many cases we checked.

At least for degree $3,$ it is possible for different sequences of functions to have the same weight recursion and the same sequence of weights, tho the functions are sometimes not affine equivalent. This happens for the functions $(1,2,4)_n$ and $(1,2,5)_n,$ for example.  It turns out that the minimal polynomial for $(1,2,5)_n$ is $x^4$ times the minimal polynomial for $(1,2,4)_n.$ The functions $(1,2,4)_7$ and $(1,2,5)_7$ are not affine equivalent since they have different distributions of the Walsh transforms \eqref{Walsh}.

\section{Appendix--Computing a quadratic  rules matrix}
We start with a quadratic RS function $f(Q_m) = f(x),$ say. We need some preliminary definitions. Though all of our definitions apply for any ordering of the functions in $Q_m$ (see \eqref{Qm}), it is convenient to assume
\begin{equation} \label{ti1}
t_1>t_2> \ldots >t_m
\end{equation}
because this gives the rules matrix in a desirable form. More details about this are in Theorem \ref{thm2} above.  

We first compute the expanded rules matrix described in Lemma \ref{expnd}. Define a sequence $y_k,~0 \leq k \leq m$ by 
$$y_0=0, y_k = \sum_{i=1}^k (t_{m-i+1}-1).$$
Define
$$X= \{2^{y_k-1}:~1 \leq k \leq m\}=\{x_i:~1 \leq i \leq m\}.$$
Define a function $F: \{subsets ~of~ X\} \rightarrow \mathbb{R},$ which maps a subset $u$ of $X$  to the sum of its elements, by
$$F(u) = \sum  x_i~\text{for all}~ x_i \in u.$$
Given a positive integer $b= \sum_{i=0}^B g_b(k) 2^k$ in base $2,$ define its $k$-th bit, using $k=0, 1, ..., B,$ counting from the right, by
$$g_b(k) = \left[ {\frac{b}{2^k}} \right] - 2 \left[{\frac{b}{2^{k+1}}}\right].$$
Define 
$$S=\{F(u):~u \in P(X)\} = F[P(X)],$$
where $P(X)$ is the power set of all subsets of $X.$
\begin{empl} \label{ex2}
Take $m=3$ and $Q_3= \{(1,4), (1,3), (1,2)\}.$ Then $t_1=4, t_2=3, t_3 = 2$ and the sequence 
$\{y_0, y_1, y_2, y_3\}$ is $\{0,t_3-1,t_3+t_2-2,t_3+t_2+t_1-3\} = \{0,1,3,6\}.$ Set $X$ is
$\{1, 4, 32\}$ and set $S$ is $\{0,1,4,5,32,33,36,37\}.$
\end{empl}

Next we build the expanded rules matrix. Let the entries in the expanded matrix of size $2^{y_m}+1$ be denoted by $e_{i,j}$ where
$0 \leq i,~j \leq 2^{y_m}.$ Then the matrix entries are defined as follows for $i,j <2^{y_m}:$
\begin{equation} \label{E1}
e_{i,j} = 1 ~ \text{if}~ g_i(y_0)= \ldots =g_i(y_{m-1})=0~\text{and}~j=\left[ {\frac{i}{2}} \right]+s~
\text{for~some}~s \in S
\end{equation}

\begin{equation} \label{E2}
e_{i,j} = 1 ~ \text{if~all}~ g_i(y_k)=1~\text{and}~j=\left[ {\frac{i}{2}} \right]-f(x)~
\text{for~some}~f(x) \in S,~|x|~\text{even} 
\end{equation}

\begin{equation} \label{E3}
e_{i,j} =-1~\text{if~all}~ g_i(y_k)=1~\text{and}~j=\left[ {\frac{i}{2}} \right]-f(x)~
\text{for~some}~f(x) \in S~|x|~\text{odd} 
\end{equation}

For the last row, $e_{2^{y_m},j}=1$ if \eqref{E3} holds and otherwise  $e_{2^{y_m},j}=0.$ For the last column, all entries are $0$ except  $e_{2^{y_m},2^{y_m}}=2.$  This is the result of the fact that the minimal and characteristic polynomials of the expanded matrix have one
integer root $2$ (see Example \ref{ex1}). 

The expanded rules matrix is then reduced in size by repeatedly deleting row $i$ and column $i,$ where row $i$ is a zero row, so that the output matrix and expanded rules matrix have the same characteristic and minimal polynomials. This process is continued until all zero rows are removed, and that final matrix is the rules matrix.  The Rules matrix is the rules matrix with the last row and last column removed.

It is important to keep in mind (as stated in Remark \ref{rk1}) that the program places exactly two numbers (each is either $1$ or $-1$)
in every row and column of the Rules matrix, and all other entries in the Rules matrix are zero, so the rules matrix is very sparse.

\end{document}